\documentclass[conference]{IEEEtran}
\IEEEoverridecommandlockouts

\usepackage{cite}
\usepackage{amsmath,amssymb,amsfonts}
\usepackage{algorithmic}
\usepackage{graphicx}
\usepackage{textcomp}
\usepackage{xcolor}
\usepackage{multirow}
\usepackage{booktabs}
\usepackage{qbordermatrix}
\usepackage{subfigure}
\usepackage{amsthm}

\newtheorem{theorem}{Theorem}  [section]

\def\BibTeX{{\rm B\kern-.05em{\sc i\kern-.025em b}\kern-.08em
    T\kern-.1667em\lower.7ex\hbox{E}\kern-.125emX}}

\begin{document}

\title{Privacy-aware Data Trading}

\author{
\IEEEauthorblockN{Shengling Wang}
\IEEEauthorblockA{\textit{\small{College of Information Science and Technology}} \\
\textit{\small{Beijing Normal University}}\\
Beijing, P.R. China \\
wangshangling@bnu.edu.cn}
\and
\IEEEauthorblockN{Lina Shi}
\IEEEauthorblockA{\textit{\small{College of Information Science and Technology}} \\
\textit{\small{Beijing Normal University}}\\
Beijing, P.R. China \\
201821210031@mail.bnu.edu.cn}
\and
\IEEEauthorblockN{Junshan Zhang}
\IEEEauthorblockA{\textit{\small{School of Electrical Computer and Energy Engineering}} \\
\textit{\small{Arizona State University}}\\
Tempe, AZ, USA \\
junshan.zhang@asu.edu}
\and
\IEEEauthorblockN{Xiuzhen Cheng}
\IEEEauthorblockA{\textit{\small{Department of Computer Science }} \\
\textit{\small{The George Washington University}}\\
Washington DC, USA \\
cheng@gwu.edu}
\and
\IEEEauthorblockN{Jiguo Yu$^*$}
\IEEEauthorblockA{\textit{\small{School of Computer Science and Technology}} \\
\textit{\small{Qilu University of Technology }}\\
Jinan, Shandong, P.R. China \\
jiguoyu@sina.com}
}

\maketitle

\begin{abstract}
The growing threat of personal data breach in data trading pinpoints
an urgent need to develop countermeasures for preserving individual
privacy. The state-of-the-art  work  either endows the data collector  with   the responsibility of data privacy or reports only
a privacy-preserving version of the data. The basic assumption of the former approach that the data collector is trustworthy  does not always hold true in reality, whereas the latter approach reduces the value of data.  In this paper, we investigate the privacy leakage issue from  the root source. Specifically, we take a fresh look to  reverse the inferior position of the data provider by making her dominate the game with the collector to solve  the dilemma in data trading.  To  that aim,  we propose the  noisy-sequentially zero-determinant (NSZD) strategies by tailoring the classical zero-determinant strategies, originally designed for the simultaneous-move game, to adapt to the noisy sequential game. NSZD strategies can empower the data provider to unilaterally set the expected payoff
of the data collector or enforce a positive relationship between  her and the data collector's expected payoffs. Both strategies can stimulate a rational data collector to behave honestly,
 boosting a healthy data trading market.
Numerical simulations are used to examine the impacts of key parameters and the  feasible region where the data provider can be an NSZD player.
Finally, we prove that  the data collector  cannot employ NSZD to further dominate the data market for deteriorating privacy leakage.

\end{abstract}

\section{Introduction}\label{intr}
We are nowhere to hide in the big-data era. From a simple click on a browser, to provide data in exchange for personalized services, all of these are exposing our habits and dispositions. In a nutshell,  personal data  has  high added values to change the way the society  thinks, lives and works, leading to a {\it personal data gold rush}, where service providers,
advertisers, and even governments are all casting a wide
net to collect personal data feverishly. This  will make  personal data  as a tradable asset \cite{senator}. However,  the personal data market is being hampered due to serious privacy issues. For example, in April 2018, a Norwegian non-profit organization called SINTEF reported the gay hookup app Grindr, which has more than 3.6 million daily active users, has been providing its users' HIV status to two other companies \cite{Grindr}; In October 2017, it was discovered that Alteryx exploded  data records for approximately 123 million U.S. households, with each record including 248 different data fields, making it simple to determine to whom the data was linked, either by looking at the details or by crosschecking with previous leaks \cite{Alteryx}.

The growing threat of personal data breach pinpoints
a need to develop countermeasures in order to protect individual
privacy. Based on who bears the responsibility of protecting data
privacy, the   state-of-the-art privacy preservation work can be classified into two  approaches: the data collector-based  and the data provider-based.
 The former \cite{collector1, collector4, collector2,collector3, provider2} was the mainstream approach   which  endows the data collector  with   the responsibility of data privacy  so that data can be protected well or
the privacy leakage will be limited to a certain level; more recent  studies, such as  \cite{provider1}, belong to the latter approach, which assumes that the data collector is not trustworthy and the data provider should  report only
a privacy-preserving version of the data to take full control of her\footnote{In this paper,  we denote the data provider as {\it she} and the collector as {\it he} for
easy differentiation. } own data privacy.

In fact, the data provider
has hitherto been at an unfair disadvantage in data trading due to the following reasons:  1) different from typical commodities,  personal data is  {\it  free commons} in nature \cite{spiekermann2015challenges}, implying that it is convenient to copy and transfer, which causes personal data easy to leak.  Further, the more powerful the digital technologies  are, the cheaper and faster the processing of sensitive data become, giving rise to more serious privacy concerns; 2) information is inherently asymmetric in  data trading. On one hand, it is hard for the data provider  to   know  {\it whether} and {\it to whom} the data collector  will resell data; on the other hand, the data collector can also mask some information to  conceal his identity for avoiding being detected.  As a result, data providers  often lack information to make privacy-sensitive decisions as empirical and theoretical research suggested \cite{acquisti2005privacy}.

The above disadvantage may make the data collector unreliable. Further, the second approach,  with the main idea of being to submit  privacy-preserving  data by its owner,  may leave  dire straits to both data collectors and providers. This is because noisy, pseudonymous, or anonymous data is clearly less valuable than the authentic data, and thus  the   data provider-based approach for privacy preservation likely would cut down economic and social benefits from big data analysis,  such as personalizing services, optimizing  decision making and predicting
future trends. 
In short, it is challenging to find a silver bullet to achieve  a win-win situation where the data provider reports true data and the collector is honest enough to protect  privacy.

In this paper, we tackle the above challenge from a fresh angle, aiming to address the privacy leakage issue from  the root source. Specifically, we study the interaction between the data provider and the data collector where both sides are allowed to take preventive measures against each other. Since the dilemma in data trading  stems from the disadvantage of the data provider, a natural idea is to  change her inferior position by making the data provider dominate the game with the collector.  To  that aim,  we resort to the
    zero-determinant (ZD) strategies \cite{press2012iterated}, whose adopter (i.e.,  the ZD player)  can unilaterally
set the expected utility  of its opponent or the ratio of their expected utilities, no matter what actions the opponent takes. The power of the ZD strategies offers  the data provider the opportunities to control the expected payoff of the data collector, based on which  incentive mechanisms can be applied for stimulating a rational data collector,  who is utility-driven,  to behave cooperatively. The main contributions of this paper can be summarized as follows:
\begin{itemize}
  \item We propose the noisy-sequentially zero-determinant (NSZD) strategies by tailoring the classical ZD strategies, originally designed for the simultaneous-move game, to adapt to the noisy sequential game, which is employed to model the data trading under consideration.
  \item The conditions that the data provider can adopt the pinning NSZD strategy are analyzed, through which the data provider is able to unilaterally set the expected payoff
of the data collector. We further study the impacts of the data provider's strategy and the noises added by both players on the expected payoff of the data collector. Moreover, numerical simulation is used to illustrate the feasible region
under which the data provider  can be
a  pinning NSZD player, which  makes  room for incentive mechanisms to stimulate the cooperation of the data collector.
  \item We study the conditions that the data provider can be an  extortionate NSZD adopter so as to enforce a positive relationship between  her and the data collector's expected payoffs. This will facilitate the cooperation of both players, because their expected payoffs can be  maximized in this case,  thus boosting a healthy data trading market.
  \item We prove that the data collector cannot be an NSZD player to employ the pinning or  extortionate  strategies. In other words, this finding reveals that the data collector, who already has an advantage, cannot take advantage of  NSZD, to further dominate the data market for deteriorating privacy leakage.
\end{itemize}



The rest of this paper is organized as follows. We introduce the related work in Section \ref{re-wo}. The game formulation for  data trading  is presented in Section \ref{game}. Section \ref{ZD} proposes the NSZD strategies. The analysis of the data provider being a  pinning NSZD and an extortionate NSZD player are presented in detail respectively, in Sections \ref{providerzd} and \ref{providerzd_ex}.    Section \ref{collectorzd} analyzes the possibility whether the data collector can employ NSZD.   The whole paper is concluded in Section \ref{conclusion}.

\section{Related work}\label{re-wo}

Based on the role of data protection, the state-of-the-art privacy preservation work can be classified into two  categories: the data collector-based and the data provider-based.

Most existing studies belong to the first category of approaches, where it is the data collector who takes full responsibility of data privacy so that   data can be protected well or  the privacy leakage  will be limited to a certain level \cite{collector1, collector4, collector2,collector3, provider2}.  For example, Haiming Jin {\it et al.} \cite{collector1} proposed an  incentive mechanism for the data collector to gather personal data, which takes advantage of  data aggregation and  perturbation to meet the satisfaction of data accuracy and privacy protection, respectively.
In \cite{collector4},   a data collector buys data from multiple
data providers and employs anonymization techniques to protect
their privacy. In this scenario,
a contract theoretic
approach was proposed for the data collector to balance  the trade-off between  privacy protection and data utility.

\cite{collector2} protects the privacy of both data and cost through designing a Bayesian incentive
compatible and privacy-preserving mechanism.
Arpita Ghosh {\it et al.}  \cite{collector3} designed a differentially private peer-prediction mechanism to collect data from privacy-sensitive population considering their  statistic of the preference for privacy, which can guarantee $\epsilon$-differential
privacy to each data provider  against the adversary who is able to observe the statistical estimate.
In  \cite{provider2}, some mechanisms are for the data collector, which  are individually rational for  data providers
with monotonic privacy valuations, truthful for those whose privacy valuations are not too large, and accurate
if there are not too many data providers with too-large privacy
valuations.  The common traits of \cite{collector2,collector3, provider2} lie in that the data collector dominates the game between it and the data provider through designing a mechanism, namely a game rule, to preserve data privacy or compensate the data providers
for their loss in privacy.


Considering that
the data provider has no control of data privacy after transferring
private data to the data collector, Weina Wang {\it et al.}  \cite{provider1} proposed
the second category of approach,  in which the data
collector is assumed to be not trustworthy and hence the data provider takes full
control of its own data privacy.
In detail,  a data provider reports only a privacy-preserving version of its data. In this scenario,   \cite{provider1} employs a game-theoretic model to derive the fundamental limits on the value of privacy, based on which the lower and upper bounds on the minimum total  payment for the data collector are deduced to guarantee a given data accuracy.


Different from the  above work, some researchers focus on  the data privacy preservation methods, rather than who will adopt them.  These methods can be classified into three kinds: the noise-based, the anonymity-based and the cryptography-based. In the noise-based methods,  randomization techniques are employed to add noises to the original data for masking some  attributes of records while keeping  the structures unchanged for future analysis \cite{noise1,noise2}. In the anonymity-based method, data generalization
is employed to hide the link between records and data owners  \cite{anonymity1,anonymity2}.   $k$-anonymity and $l$-diversity are the two well-known anonymity-based methods.  As the key technology of data security, cryptography-based methods convert plain texts into cipher ones through encryption keys or algorithms to prevent an illegal user from getting useful information \cite{cryptography1,cryptography2}.


\section{Game formulation for data trading}\label{game}

Personal data can optimize decision making and promote economic
growth. It is a valuable asset in the market,  where the data collector purchases informative data from the data provider, which  can be an individual or a data broker who  aggregates information from a variety of sources. Since data is  free commons in nature and there is often asymmetric information in data trading,  data breach becomes serious.  In this scenario, the data provider should decide to  submit either a  fully identifiable data or its  privacy-preserving version such as noise-injected data, which is a dilemma due to following trade-off:  the unprocessed data is more valuable but   jeopardizes the owner's privacy; the data collector should determine whether to resell the provider's
personal data to third parties:  the sale of such data may
gain more profit for the data collector but would do  harm to his reputation
 when detected, which may reduce his future chances of data trading.

Thus motivated, we cast the above problem in a game-theoretic setting, where  the  data trading game between the data provider and the collector is a sequential one in that the data  provider moves first (i.e., submitting the data) and the collector makes his decision
later (i.e., whether resells the data). As illustrated in Fig. 1,  both players, namely the data provider and the collector, can take one of two actions: the cooperation $C$ or the defection $D$. In detail, the data provider takes action $C$ if  she reports original data or $D$ if provides data with noise injection. Similarly, the data collector is cooperative ($C$) when he protects data privacy or detective ($D$)  if it resells the data to third parties.  In this paper, we assume that the data collector is intelligent to conceal his identity, i.e., adding some noises to the data for interfering the data provider's judgement,  before leaking private data. Thus,
the state vector of the data trading game is   $[xy]_4$, where $x \in \{C, D\}$ and $y \in \{C, D\}$ respectively denote the actions of the data provider and the collector.

\begin{figure}\label{game1}
\includegraphics[scale=0.36]{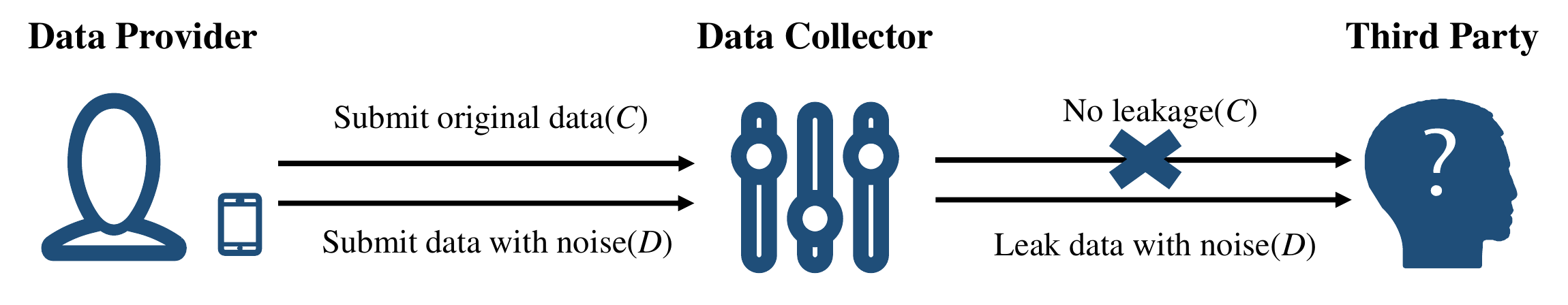}
\caption{Data trading game  between the data provider and the data collector.}
\end{figure}

Table~\ref{tab:payoffs} summarizes the payoff matrix under the above four states, where the data provider is denoted as X and the collector is indicated as Y.  In this payoff matrix, $\tau_1$, $\tau_2$, $\sigma_1$ and $\sigma_2$ are all scale parameters. $e_1$ and $e_2$ are noises respectively added by the data provider and the collector which are normalized within $[0,1]$. Since  $e_1$ is employed to preserve data privacy,  the value of the data is proportional to $1-e_1$. Moreover,  because $e_2$ is adopted to mask the identity of the data collector, $1-e_2$ reflects the probability that the malicious behavior of  the data collector can be detected.
Other parameters in the payoff matrix are outlined as follows: 1)  $C_P$ and $C_C$ are  respectively the  profits of the data provider and the collector obtained
from the data trading when two players are cooperative and thus their profits will reduce to $(1-e_1)C_P$ and $(1-e_1)C_C$ respectively if the data provider defects;  2)  the data collector can obtain an extra profit
$C_{C1}$ from leaking the data provider's information while the data provider would have a loss of $C_{P1}$ in this situation; and thus when the noise-added data is traded, such a profit and a loss will change to $(1-e_1)C_{C1}$ and $(1-e_1)C_{P1}$  respectively; 3) when the  malicious behavior of the data collector  is perceived, he will suffer a loss of $C_{C2}$  due to the reputation loss;  at the same time,  the data provider  would   gain a payoff of $C_{P2}$ because she makes up the future data breach loss. In light of the above analysis, the probability of detecting privacy leakage is $1-e_2$.

Based on the above analysis, we can deduce the payoffs of two players in each state. For example, when both players are defective, namely in the state of $DD$, the payoff of the data provider includes three parts: the profit obtained from selling noise-added data to the data collector (i.e., $(1-e_1)C_P$),  the expected loss of privacy leakage (namely $(1-e_1)C_{P1}$)   and the expected profit ($(1-e_2)C_{P2}$) due to detecting the malicious behavior of the data collector which can compensate for future privacy leakage loss. Hence, the payoff of the data provider is $(1-e_1)C_P-(1-e_1)C_{P1}+(1-e_2)C_{P2}$  as shown in Table~\ref{tab:payoffs}. We omit the analysis on how to calculate the payoffs of two players in other states  for  avoiding redundancy.

\begin{table*}
  \caption{Payoff matrix in the data trading game.}
  \label{tab:payoffs}
  \centering
  \begin{tabular}{cccc}
    \toprule
    X's Action &Y's Action & X's Payoff &Y's Payoff \\
    \midrule
    \multirow{2}*{$C$}& $C$& $C_P$ &$C_C$\\
    ~& $D$ & $C_P-C_{P1}+(1-e_2)C_{P2}$ &$C_C+C_{C1}-(1-e_2)C_{C2}$\\
    \cmidrule{2-4}
     \multirow{2}*{$D$}& $C$& $(1-e_1)C_P$ &$(1-e_1)C_C$\\
    ~& $D$ & $(1-e_1)C_P-(1-e_1)C_{P1}+(1-e_2)C_{P2}$ &$(1-e_1)C_C+(1-e_1)C_{C1}-(1-e_2)C_{C2}$\\
    \bottomrule
  \end{tabular}
\end{table*}

Denote   the payoff vectors of the data provider and the collector  as $$\bold{U_P}=(U_P(CC), U_P(CD), U_P(DC), U_P(DD)),$$ and  $$\bold{U_C}=(U_C(CC), U_C(CD), U_C(DC), U_C(DD)),$$ respectively. In the following, we will analyze the relationship among these payoffs. From the perspective of the data provider, her payoff when the data collector acts $C$ is always higher than that when $D$ is adopted by the data collector, implying an honest data collector is good for the data provider. In addition, the data provider can gain more profits if she chooses  $C$ rather than $D$ when the data collector is cooperative, which accords with the practice that the original data is more valuable  than the noisy data in a healthy data trading environment. As a result, we can obtain
\begin{equation}\label{eq:p1}
U_P(CC)>U_P(CD),
\end{equation}
\begin{equation}\label{eq:p2}
U_P(CC)>U_P(DC)>U_P(DD).
\end{equation}

If $C_P> C_{P1}$, which means that the value of data is larger than that of the privacy, and hence the data provider will obtain the least profit when both players choose $D$. Otherwise, when $C_P<C_{P1}$,  the data provider is more sensitive to data privacy, which makes   her payoff minimized when she acts $C$ while the data collector takes the action $D$.

We elaborate further on (1) and (2). From the viewpoint of the data collector, we make an assumption that  the data collector can always obtain higher when he defects than when he  cooperates. This is because if the data collector cannot benefit from leaking privacy, his dilemma would disappear and he would decide not to resell data to third parties. In this case, it is not of interest to study the privacy leakage problem in data trading. Hence, the above assumption clarifies the conditions under which our problem exists. Moreover, we think the payoff of  the data collector when the provider acts $C$   is more  than  that when she is defective. In another word, the condition guarantees the noisy data is less valuable than the original data. Conclusively, we have
\begin{equation}\label{eq:c1}
U_C(CD)>U_C(CC)>U_C(DC),
\end{equation}
\begin{equation}\label{eq:c2}
U_C(CD)>U_C(DD)>U_C(DC).
\end{equation}

Based on \eqref{eq:c1} and \eqref{eq:c2},  the data collector can maximize and minimize his payoff when the states are $CD$ and $DC$ respectively.

\section{Noisy-sequentially zero-determinant strategies}
\label{ZD}

\begin{figure*}
\caption{Markov matrix of the data trading game.}
\label{fig:transition}
\begin{small}
\begin{displaymath}
\setlength{\arraycolsep}{1.5pt}
\centering
\left[
\begin{array}{cccc}
p_1q_1& p_1(1-q_1)& (1-p_1)[(1-e_1)q_1+e_1q_2]&(1-p_1)[(1-e_1)(1-q_1)+e_1(1-q_2)]\\
\left[e_2p_1+(1-e_2)p_2\right]q_1& \left[e_2p_1+(1-e_2)p_2\right](1-q_1)&
\left(\begin{array}{c}\left[e_2(1-p_1)+(1-e_2)(1-p_2)\right]\\\left[(1-e_1)q_1+e_1q_2\right]\end{array}\right)&
\left(\begin{array}{c}\left[e_2(1-p_1)+(1-e_2)(1-p_2)\right]\\\left[(1-e_1)(1-q_1)+e_1(1-q_2)\right]\end{array}\right)\\
p_3q_1& p_3(1-q_1)& (1-p_3)[(1-e_1)q_1+e_1q_2]&(1-p_3)[(1-e_1)(1-q_1)+e_1(1-q_2)]\\
\left[e_2p_3+(1-e_2)p_4\right]q_1& \left[e_2p_3+(1-e_2)p_4\right](1-q_1)&
\left(\begin{array}{c}\left[e_2(1-p_3)+(1-e_2)(1-p_4)\right]\\\left[(1-e_1)q_1+e_1q_2\right]\end{array}\right)&
\left(\begin{array}{c}\left[e_2(1-p_3)+(1-e_2)(1-p_4)\right]\\\left[(1-e_1)(1-q_1)+e_1(1-q_2)\right]\end{array}\right)\\
\end{array}
\right]
\end{displaymath}
\end{small}
\end{figure*}

As mentioned above,  due to being  at a disadvantage,   the data provider may have to sell a privacy-preserving version of data to the collector.
Hence, to realize a win-win situation where the data provider reports authentic data and the collector preserves  privacy, we take a fresh approach to reverse  the disadvantaged position of the data provider. To that aim,  we empower the data provider to adopt
ZD strategies \cite{press2012iterated}, so that she can  unilaterally set  the utility of the collector to  motivate  a healthy data trading environment. However, the original ZD strategies cannot be directly employed here, because the interaction between the data provider and the collector is a noisy sequential game. Particularly, in our scenario,   the noises $e_1$ and $e_2$ are added into data by the  provider and the collector respectively, leaving it to the player to determine  the action of the opponent only through observations. Due to the errors in observation, the calculation of the state transition probabilities should take into considerations all possible actions taken by the opponent, which is in stark contrast from that in the original ZD strategies.

To tailor the ZD strategies to the data trading game, we propose {\it the noisy-sequentially zero-determinant (NSZD) strategies}, which are detailed as follows. In the data trading game, the data provider  takes  action first, and thus her strategy is made based on  the joint actions of two players in the previous round. Let $\mathbf p$ be the strategy of the data provider and then we have $$\mathbf p=(p_1,p_2,p_3,p_4).$$ The elements in $\mathbf p$   are the cooperation probabilities of the data provider under each of the previous outcome $Cg, Cb, Dg$ and $Db$, where $g$ and $b$ represent observation results reflecting the data collector are cooperative and defective respectively.

Since the data collector is a follower in the sequential game,   he can  make the decision according to his observation on the actions of the data provider in the current round, which is denoted as $$\mathbf q=(q_1,q_2).$$ In  $\mathbf q$,   each element denotes the cooperation probabilities of the data collector under each of his observations $g$ and $b$.

With the definitions of $\mathbf p$ and $\mathbf q$, the Markov matrix of the data trading game can be constructed as $$\mathbf M=[M_{vw}]_{4\times4},$$ where $M_{vw}$ means the transition probability from state $v$ to $w$. Both $v$ and $w$ belong to the state space $\{CC,CD,DC,DD\}$, so $\mathbf M$ is a $4\times4$ matrix. The state transition probability can be obtained by $$M_{vw}=F_{vw}\cdot G_{vw},$$ of which $F_{vw}$ and $G_{vw}$ respectively represent the probabilities of actions made by the data provider and the data collector to form state $w$ when the state of the previous round is $v$. $F_{vw}$ and $G_{vw}$ can be calculated as follows.


\begin{displaymath}
F_{vw}=
\left\{
\begin{aligned}
f(1)~~~~~~~~~  ~v=CC\\
e_2f(1)+(1-e_2)f(2) ~v=CD\\
f(3)~~~~~~~~~  ~v=DC\\
e_2f(3)+(1-e_2)f(4) ~v= DD\\
\end{aligned}
\right.,
\end{displaymath}
where $f(i)={p_i}^{\alpha_i}(1-{p_i})^{1-\alpha_i}$, $i \in \{1,2,3,4\}$ and when the data provider chooses to take action $C$, $\alpha_i=1$, otherwise $\alpha_i=0$.
 \begin{displaymath}
G_{vw}=
\left\{
\begin{aligned}
q_1~~~~~~~~~~~~~w=CC\\
1-q_1~~~~~~~~~~w=CD\\
(1-e_1)q_1+e_1q_2 ~~~~~w=DC\\
(1-e_1)(1-q_1)+e_1(1-q_2) ~w=DD\\
\end{aligned}
\right..
\end{displaymath}

We further  explain the calculation of  $F_{vw}$ and $G_{vw}$.  If the opponent cooperates, the player   can judge its action clearly since  no noise is added to interfere with the observation. In another word,   the action of an opponent obtained through observation is the same as the actual one in this case. Therefore, when $v=CC, DC$,  implying that the data collector is cooperative  in the previous round,  the data provider gets the observation $g$ with the probability of 1,  and then $F_{vw}=f(1)$ in the case $v=CC$ and $F_{vw}=f(3)$ in the case $v=DC$ because in these two situations, the data provider adopts her strategy only based on the previous outcome $Cg$ and $Dg$, respectively. Likewise, when $w=CC, CD$, implying that the data provider cooperates in the current round, the data collector observes $g$ with no interference and takes action under the current outcome $g$. Thus, $G_{vw}=q_1$ in the case $w=CC$ and $G_{vw}=1-q_1$ in the case $w=CD$.

 When the opponent takes action $D$, observation errors may occur whose chances are proportional to the noises $1-e_1$ and $e_2$. In this case, the transition probability depends on the probabilities of all possible observations. To explain how to calculate the transition probability  in this situation, we take $v=DD, w=DD$ as an example to deduce $F_{vw}$ and $G_{vw}$.
  As shown in  Fig. \ref{fig:tran-example},  when the data collector defects in the previous round,
 the data provider may get a false observation with a probability of  $e_2$   and then she has the probability of ($1-p_3$) to take action $D$ under the previous outcome $Dg$. On the contrary, the data provider has the probability of ($1-e_2$) to know the data collector's defective action and chooses to defect with the probability of ($1-p_4$). Therefore, the probability of the data provider taking action $D$ when $v=DD$ is $[e_2(1-p_3)+(1-e_2)(1-p_4)]$.
    When the data provider takes the action $D$ with noise $e_1$, the data collector considers the opponent is cooperative and defective  with  the probability of $(1-e_1)$  and $e_1$ respectively. Hence, the data collector would have the probabilities of $(1-q_1)$ and $(1-q_2)$ to choose the action $D$ under these two circumstances. In this case, the probability that the data collector takes action $D$ when the data provider adopts $D$ is $[(1-e_1)(1-q_1)+e_1(1-q_2)]$.
     Taking advantage of the similar methods, we can
calculate the transition probabilities in other cases, and thus the state transition matrix $\mathbf{M}$, also known as the Markov matrix, can be deduced as  shown in Fig.~\ref{fig:transition}.


\begin{figure}
\includegraphics[height=1.476in, width=3.5in]{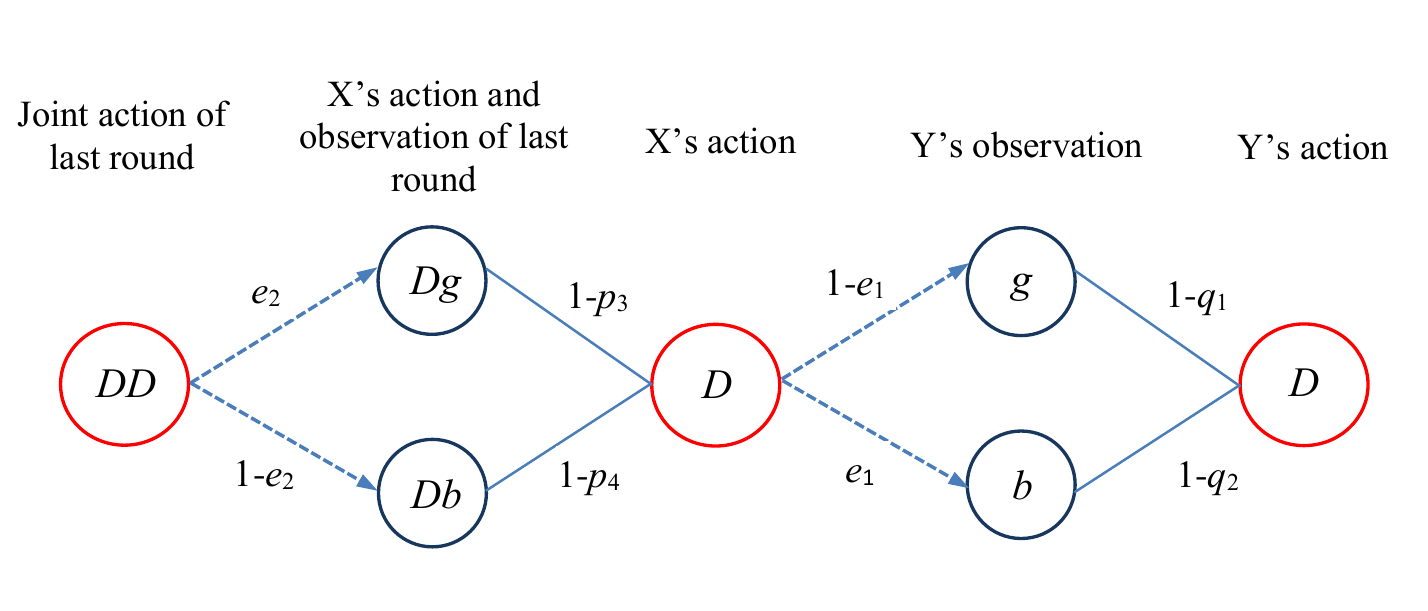}
\caption{Illustration of the transition from previous state $DD$ to a new state $DD$. The red circle shows the real action of two players while the blue one is the combination of action and observation or just observation.}
\label{fig:tran-example}
\end{figure}

Let $\mathbf v$ be the stationary vector of $\mathbf M$ such that
\begin{equation}\label{eq:ZD1}
\mathbf v^\mathrm{T}\mathbf M=\mathbf v^\mathrm{T} or\ \mathbf v^\mathrm{T}\mathbf M'=\mathbf 0,
\end{equation}
where $\mathbf M'=\mathbf M-\mathbf I$ and $\mathbf 0$ is a vector with each element being 0. According to Cramer's rules, the following equation  holds
\begin{equation}\label{eq:ZD2}
Adj(\mathbf M')\mathbf M'=det(\mathbf M')\mathbf I= \mathbf 0,
\end{equation}
in which $Adj(\mathbf M')$ is the adjugate matrix of $\mathbf M^{'}$. In light of \eqref{eq:ZD1} and \eqref{eq:ZD2}, we know that the stationary vector $\mathbf v$ is proportional to every row of $Adj(\mathbf M')$ \cite{press2012iterated}. Consequently, $\mathbf v$ can equal to the fourth row of $Adj(\mathbf M')$, each element of which is calculated from the first three columns of $\mathbf M'$.  The determinant  of $\mathbf M'$  remains unchanged if we take the following two elementary transformations: i) adding the first column of $\mathbf M'$ into the second one and ii) adding the second column with multiplier factor [$(1-e_1)q_1+e_1q_2$] into the third one. Finally, replacing the last column of $\mathbf M'$ by an arbitrary four-element vector $\mathbf f=[f_1, f_2, f_3, f_4]$, we can get a new matrix
\\\\
$ \mathbb{\small{M}}= \bordermatrix*[||]{%
...&{\color{red} p_1-1} &{\color{blue}0} & f_1  &\cr
...&{\color{red}e_2p_1+(1-e_2)p_2-1} &{\color{blue}0} &f_2  & \cr
...&{\color{red}p_3} &{\color{blue}(1-e_1)q_1+e_1q_2-1} &f_3 & \cr
...&{\color{red}e_2p_3+(1-e_2)p_4} &{\color{blue}(1-e_1)q_1+e_1q_2} &f_4 & \cr\\
 &\mathbf {\hat p} &\mathbf {\hat q} &  & \cr
}$
\\\\
with  the second column (denoted as $\mathbf {\hat p}$) being solely controlled by the data provider and the third column (denoted as $\mathbf {\hat q}$) being only related to the strategies of the data collector. Moreover, we omit the first column because what we focus on is the relation between the second (or third) column and the fourth column. The dot product of the stationary vector $\mathbf v$ and $\mathbf f$ can be written as
\begin{displaymath}
\mathbf v \cdot \mathbf f=D(\mathbf {\hat p},\mathbf {\hat q}, \mathbf f),
\end{displaymath}
because the determinant $D(\mathbf {\hat p},\mathbf {\hat q}, \mathbf f)$ can be expanded by the minors on the fourth column of $\mathbb{\small{M}}$, and each coefficient of $f_c$ $(c\in\{1,2,3,4\})$ is just the components of $\mathbf v$ described above. We have known that the stationary vector of a Markov matrix is the probability vector in the long run. Therefore, if  $\mathbf f = \bold{U_P}$, the payoff vector of  the data provider,   the determinant $D(\mathbf {\hat p},\mathbf {\hat q}, \mathbf U_P)$  is the expected payoff of the data provider. Follow the same way, we can obtain the data collector's expected payoff when we replace $\mathbf f$ by $\bold{U_C}$.

Based on \cite{press2012iterated}, the normalized expected payoff of the data provider ($S_P$) and the data collector ($S_C$) can be respectively derived as:
\begin{displaymath}
S_P=\frac{\mathbf{v\cdot U_P}}{\mathbf {v \cdot 1}}=\frac{D(\mathbf {\hat p}, \mathbf {\hat q}, \mathbf{U_P})}{D(\mathbf {\hat p}, \mathbf {\hat q},\mathbf 1)},
\end{displaymath}
\begin{displaymath}
S_C=\frac{\mathbf{v\cdot U_C}}{\mathbf {v \cdot 1}}=\frac{D(\mathbf {\hat p},\mathbf {\hat q},\mathbf{U_C})}{D(\mathbf {\hat p}, \mathbf {\hat q},\mathbf 1)}.
\end{displaymath}
Thus, the liner relation of the two expected payoffs with coefficients $\alpha \in \textbf{R},   \beta \in \textbf{R}, \gamma \in \textbf{R}$ can be written as:
\begin{equation}
\alpha S_P+\beta S_C+\gamma=\frac{D(\mathbf {\hat p},\mathbf {\hat q},\alpha \mathbf{U_P}+\beta \mathbf U_C+\gamma \mathbf 1)}{D(\mathbf {\hat p},\mathbf {\hat q},\mathbf 1)}.
\end{equation}
If the data provider  chooses a proper strategy  satisfying
\begin{equation}\label{cdpzd}
\mathbf {\hat p}=\alpha \mathbf{U_P}+\beta \mathbf U_C+\gamma \mathbf 1,
\end{equation}
 or the data collector's strategy can meet
 \begin{equation}\label{cdczd}
\mathbf {\hat q}=\alpha \mathbf U_P+\beta \mathbf U_C+\gamma \mathbf 1,
\end{equation}
 $\alpha S_P+\beta S_C+\gamma=0$, indicating that the linear relation between the two expected payoffs ($S_P$ and $S_C$) is established. Since the strategies referred above are realized by setting the determinant to zero and they function in a noisy sequential game,  we call them  {\it the noisy-sequentially zero-determinant (NSZD) strategies}. In the following, we will detail how to use   the NSZD strategies for the data provider to set the  expected utility of the data collector to create a win-win situation where the data trading market is healthy.

\section{The data provider as a pinning NSZD player}
\label{providerzd}

In this study, we will analyze how the  data provider  takes advantage of the NSZD strategies to dominate the game with the data collector.  In detail, the data provider can adopt the pinning and extortionate NSZD strategies. By the former strategy,  the data provider can unilaterally control the payoff of the opponent in the noisy sequential game no matter how the collector will react, while through the latter strategy, the
data provider can enforce an extortionate linear relation  between her and the collector's  expected payoff, thereby always seizing an advantageous share of payoffs.  In the following, the pinning NSZD strategy is detailed.


\begin{theorem}
\label{theorem:dcpinning}
The data provider can unilaterally set the expected payoff of the opponent (i.e., the data collector) as
\begin{equation}\label{relation}
S_C=-\frac{\gamma}{\beta}=\frac{A(1-p_1)+Bp_4}{(1-p_1+p_4)},
\end{equation}
where $A=\frac{U_C(DD)-e_2U_C(DC)}{1-e_2}$ and $B=U_C(CC)$.
\end{theorem}
\begin{proof}
In light of \eqref{cdpzd}, if the data provider can find proper $p_1, p_2, p_3$ and $p_4$ to let $\mathbf {\hat p}=\beta \mathbf{U_C}+\gamma\mathbf 1$ (set $\alpha =0$) hold, then she can unilaterally set the opponent's expected payoff as
\begin{equation}\label{eq:pinning1}
S_C= -\frac{\gamma}{\beta}.
\end{equation}
Particularly,  $\mathbf {\hat p}=\beta \mathbf{U_C}+\gamma\mathbf 1$  equals to that the following equation set holds,
\begin{equation}\label{eqset1}
\left\{
\begin{aligned}
p_1-1=\beta U_C(CC)+\gamma\\
e_2p_1+(1-e_2)p_2-1=\beta U_C(CD)+\gamma\\
p_3=\beta U_C(DC)+\gamma\\
e_2p_3+(1-e_2)p_4=\beta U_C(DD)+\gamma\\
\end{aligned}
\right..
\end{equation}
There are six variables ($p_1,p_2,p_3,p_4,\beta$ and $\gamma$) in the above equation set. We use $p_1$ and $p_4$ to represent the rest variables. Then $p_2$ and $p_3$ can be written as
\begin{equation}\label{eq:pinning3}
\begin{aligned}
p_2=&\frac{1}{D_1}\{[U_C(CD)-U_C(DD)+e_2(U_C(DC)\\&-U_C(CC))]p_1+[U_C(CC)-U_C(CD)](1+p_4)\},
\end{aligned}
\end{equation}
\begin{equation}\label{eq:pinning4}
\begin{aligned}
p_3=&\frac{1}{D_1}\{[U_C(DD)-U_C(DC)](1-p_1)+[U_C(CC)\\&-U_C(DC)](1-e_2)p_4\},
\end{aligned}
\end{equation}
where $D_1=U_C(CC)-U_C(DD)-e_2[U_C(CC)-U_C(DC)]$. Then,  according to \eqref{eqset1}, \eqref{eq:pinning3},  and \eqref{eq:pinning4},  the expected payoff of the data collector can be represented as \eqref{relation}.
\end{proof}
\begin{theorem}
\label{theorem:ipp}
When  $p_4\neq0,  p_1\neq1$, the expected payoff of the data collector $S_C$ increases as $p_1$ or $p_4$ increases.
\end{theorem}
\begin{proof}
In light of \eqref{relation}, we found  i) when $p_4=0$, $S_C=A$;  ii) when $p_1=1$, $S_C=B$; and iii) when  $A=B$,  $S_C=A=B$.
For further analyzing  the impact of the data provider's strategy  on the expected payoff of the collector, we deduce the following  partial differential equations for $p_1$ and $p_4$. These are,
\begin{equation}\label{eq:partial}
\begin{aligned}
\frac{\partial {S_C}}{\partial {p_1}}&=\frac{(B-A)p_4}{(1-p_1+p_4)^2}\\
\frac{\partial {S_C}}{\partial p_4}&=\frac{(B-A)(1-p_1)}{(1-p_1+p_4)^2}
\end{aligned}.
\end{equation}

According to \eqref{eq:c1} and \eqref{eq:c2}, $A<B$ should hold in the data trading game since the more possible the data provider will cooperate, the larger the data collector's expected payoff.
 Therefore,    when  $p_4\neq0,  p_1\neq1$, $S_C$ increases as $p_1$ or $p_4$ increases
in light of \eqref{eq:partial}.
\end{proof}

Next, we will discuss the impacts of the noises $e_1$ and $e_2$, two key parameters in the data trading game. $e_1$ and $e_2$ are closely related to the strategies of the data provider and the collector since they  directly influence the payoff vectors of both players and the transfer probabilities (the Markov matrix).
Through our analysis, we have the following theorem:
\begin{theorem}
\label{theorem:ipe}
The noises $e_1$ and $e_2$  respectively added by the data provider and the collector, have  the opposite impacts on the expected payoff of the data collector $S_C$. In detail,  $S_C$ decreases  monotonously with  $e_1$ and increases monotonously with $e_2$.
\end{theorem}
The above theorem can be easily obtained by using $S_C$ to derive partial derivatives for $e_1$ and $e_2$. That is,
\begin{equation}\label{eq:e1}
\frac{\partial {S_C}}{\partial {e_1}}=-\frac{1-p_1}{1-p_1+p_4}\frac{(1-e_2)C_C+\tau_1 C_{C1}}{1-e_2}<0,
\end{equation}
\begin{equation}\label{eq:e2}
\frac{\partial {S_C}}{\partial {e_2}}=\frac{1-p_1}{1-p_1+p_4}\frac{(1-e_1)\tau_1C_{C1}}{(1-e_2)^2}>0.
\end{equation}
 According to Theorem \ref{theorem:ipe},    the larger the noise $e_1$, the smaller the expected payoff of the data collector $S_C$, while  a small noise $e_2$ will lower the expected payoff of the data collector.
These accord with the reality that the larger noise added by the data provider will reduce the value of her data, which leads to a lower payoff of the data collector, and when the data collector reduces his noise, his malicious behavior is easier  to be detected, damaging his reputation and hence causing a less expected payoff of his own.

Finally, we employ the numerical analysis\footnote{It is worth noting that in all numerical simulations of our paper,  we test multiple parameter
settings, but we only show a part of them since other results   present similar trends.  So we omit them for
avoiding redundancy. } to study the impact of the data provider's strategy on $S_C$, the expected payoff of the data collector. Figs. \ref{fig:pinning1-1} and  \ref{fig:pinning1-2} respectively   show how $S_C$ changes with the strategy of the data provider under different noise and profit settings. Both Figs. \ref{fig:pinning1-1} and  \ref{fig:pinning1-2} verify our above analysis that the increasing of $p_1$ and $p_4$ leads to that of $S_C$; moreover, the data collector gets a lower payoff when $e_1$ increases or $e_2$ decreases, and vice versa. In addition,  the red part in each subfigure indicates  the feasible region under which the data provider can control the data collector's expected payoff, demonstrating that the data provider can be a pinning NSZD player with  proper strategies.

\begin{figure*}
\centering
\subfigure[$e_1=0.3$, $e_2=0.5$]
{ \label{fig:subfig:a}
\includegraphics[height=1.47in, width=2.2in]{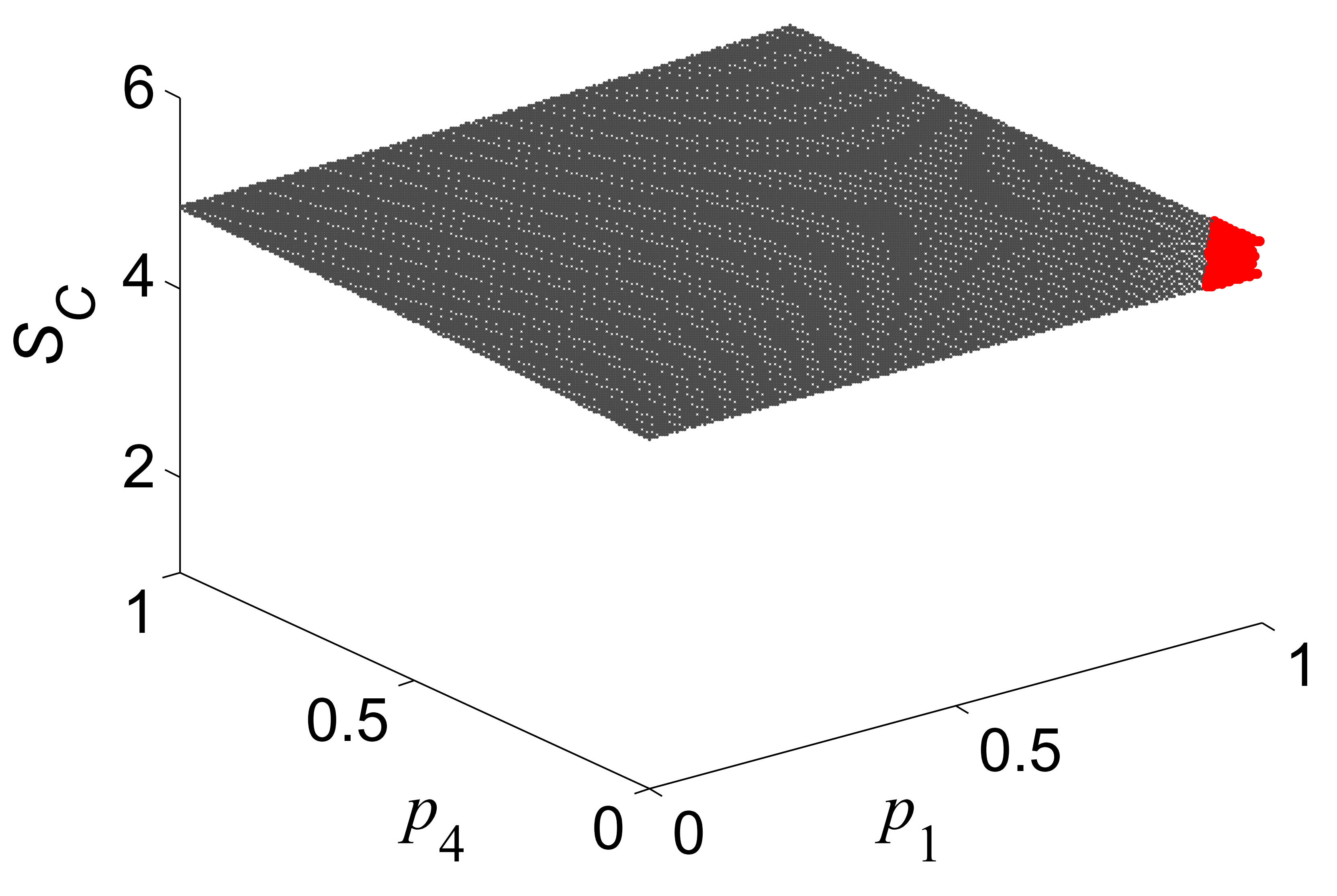}}
\subfigure[$e_1=0.5$, $e_2=0.5$]
{ \label{fig:subfig:b}
\includegraphics[height=1.47in, width=2.2in]{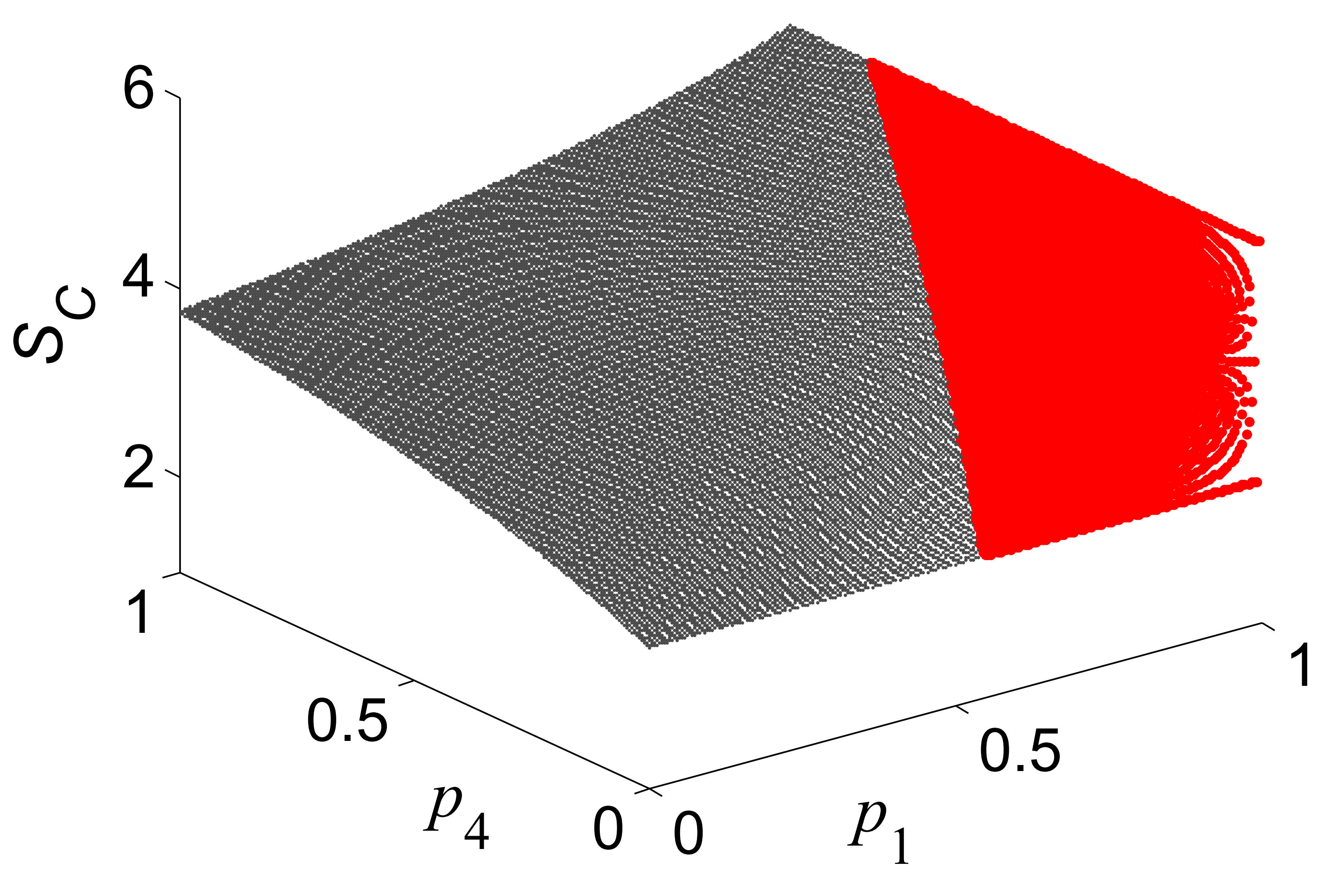}}
\subfigure[$e_1=0.5$, $e_2=0.3$]
{ \label{fig:subfig:c}
\includegraphics[height=1.47in, width=2.2in]{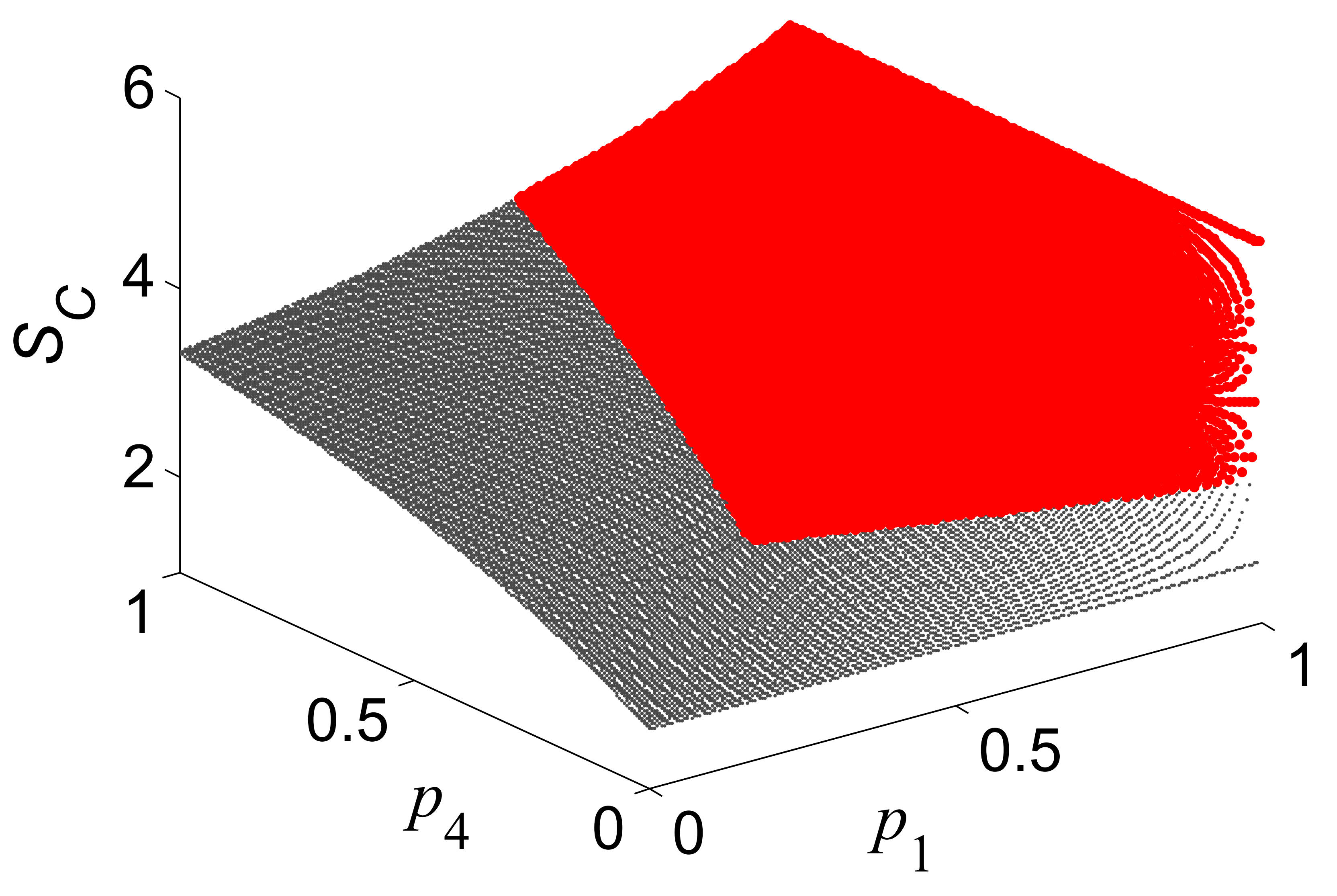}}
\caption{Impact of the data provider's strategy on the expected payoff of the data collector with $C_P=C_C=5$, $C_{P1}=C_{C1}=2$, $C_{P2}=C_{C2}=3$.} \label{fig:pinning1-1}
\end{figure*}

\begin{figure*}
\centering
\subfigure[$e_1=0.3$, $e_2=0.5$]
{ \label{fig:subfig:a}
\includegraphics[height=1.47in, width=2.2in]{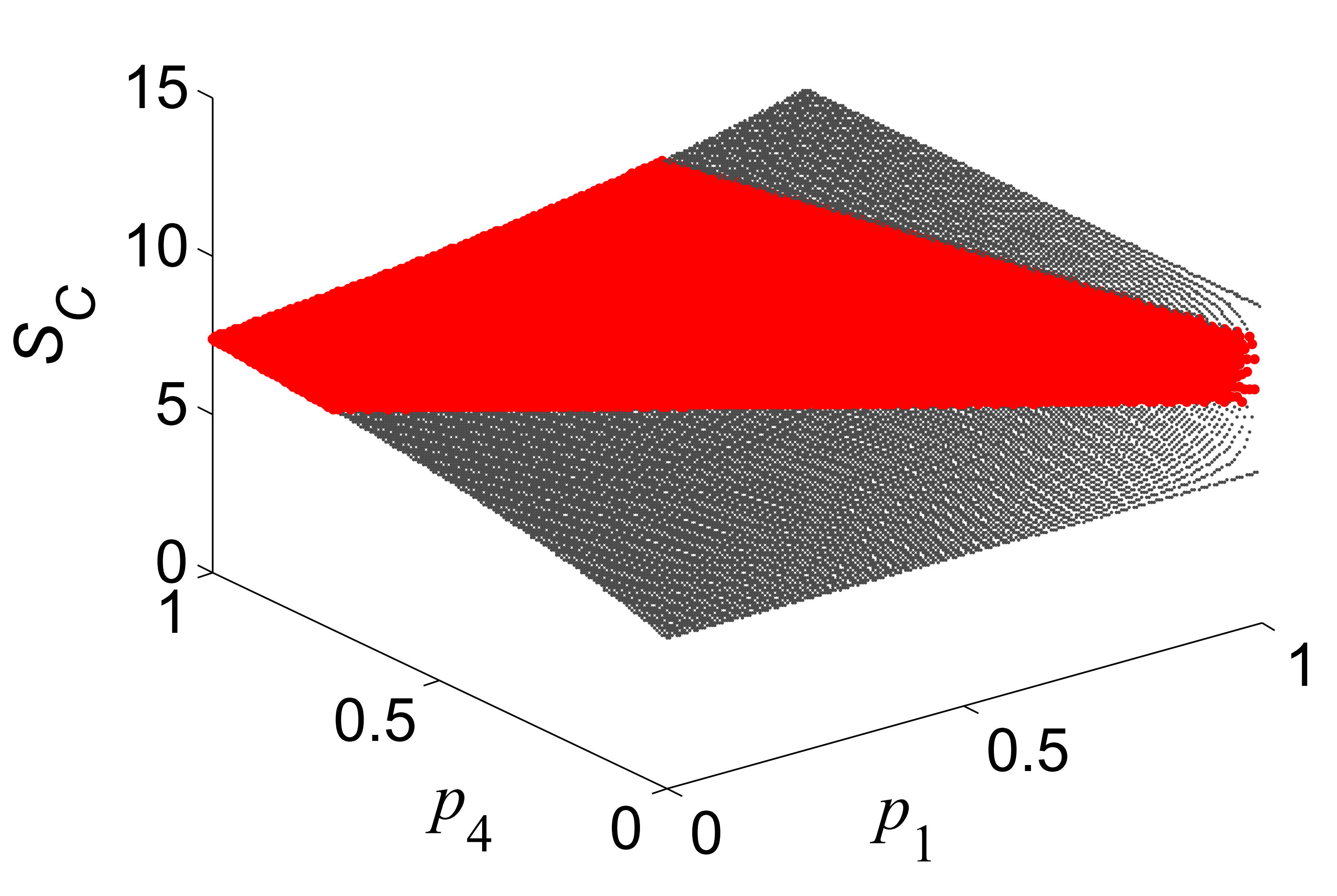}}
\subfigure[$e_1=0.5$, $e_2=0.5$]
{ \label{fig:subfig:b}
\includegraphics[height=1.47in, width=2.2in]{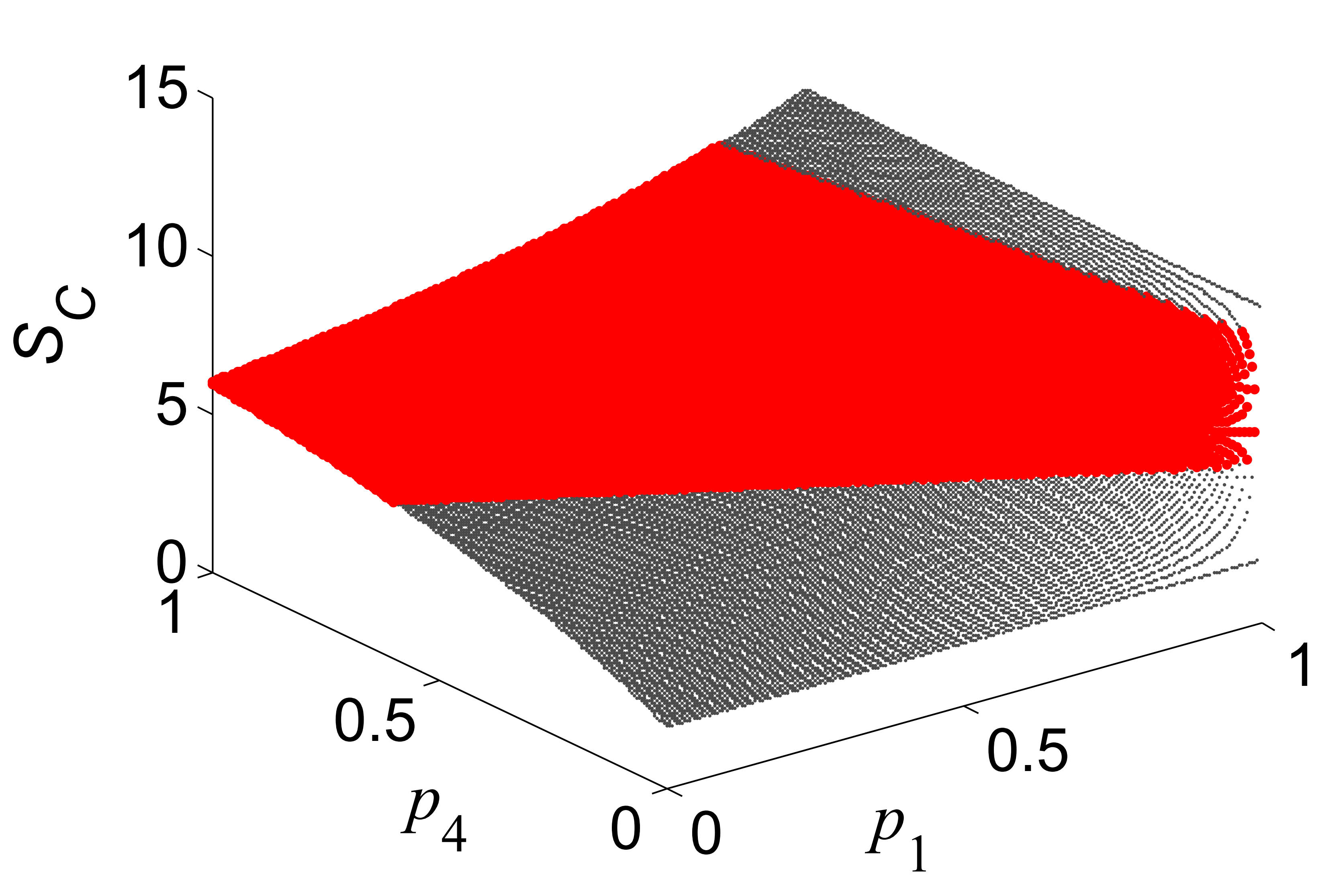}}
\subfigure[$e_1=0.5$, $e_2=0.3$]
{ \label{fig:subfig:c}
\includegraphics[height=1.47in, width=2.2in]{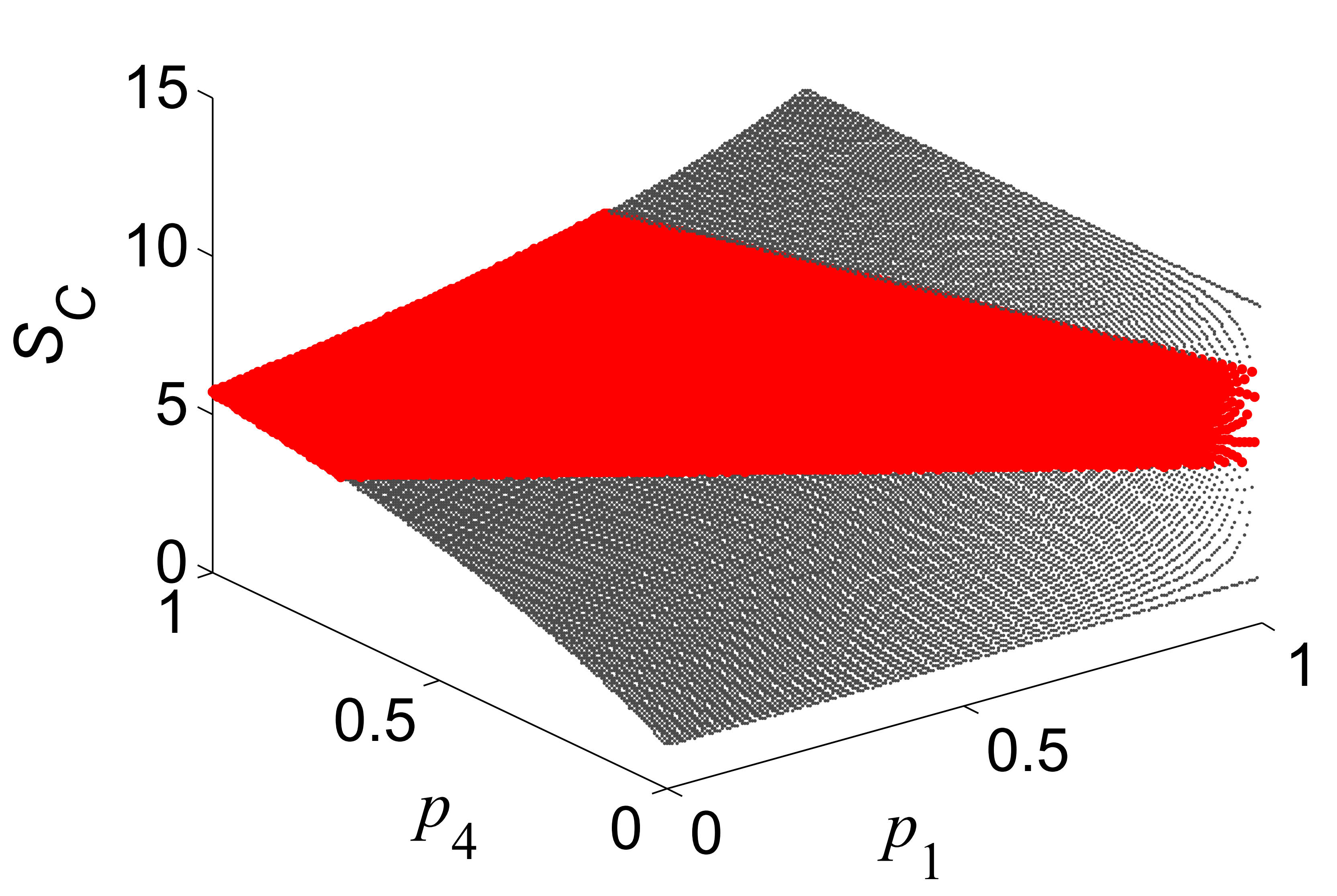}}
\caption{Impact of the data provider's strategy on the expected payoff of the data collector with $C_P=C_C=10$, $C_{P1}=C_{C1}=2$, $C_{P2}=C_{C2}=5$.} \label{fig:pinning1-2}
\end{figure*}


Being a pinning NSZD player, the data provider can unilaterally set the data collector's expected payoff. When the data collector's strategy is driven by maximizing his expected payoff, he will prefer to  the strategy that can  increase his expected payoff. Hence, taking advantage of the pinning NSZD strategy, the data provider can increase the data collector's expected payoff when he chose to cooperate and decrease his expected payoff when he defected in the previous round to incentivize the  cooperation of the data collector. Such an incentive mechanism  with the reward-punishment method is intuitive and has been well discussed in  \cite{hu2017anti}. In \cite{hu2017anti}, we formulated the interaction between the requestor and any worker in a crowdsourcing scenario  as a  prisoner's dilemma game and then  employed the ZD strategy to promote the cooperation of any malicious worker. Although the scenario of \cite{hu2017anti} is different from ours  in which the data trading  is modeled as a noisy sequential game,  once the data provider can be a pinning NSZD player, the pinning NSZD-based incentive mechanism is similar to that in \cite{hu2017anti}.  Therefore, we omit the incentive scheme to avoid redundancy.  Instead, we mainly focus on under which conditions  the data provider can adopt the pinning  NSZD strategy to set  the data collector's expected payoff, as we analyzed above.


\section{The data provider as an extortionate NSZD player}
\label{providerzd_ex}

The above pinning NSZD strategy focuses on unilaterally setting the expected payoff of the opponent while the extortionate one aims at establishing  an extortionate linear relation between
the  extortionate NSZD player and its opponent's expected payoffs. In light of \cite{press2012iterated}, \eqref{cdpzd} can be rewritten as the following form
\begin{equation}\label{eq:extortion1}
\mathbf {\hat p}=\phi[(\mathbf U_P-l_1\mathbf 1)-\chi (\mathbf U_C-l_2\mathbf 1)],
\end{equation}
where $\phi, \chi >1$ and $l_1 >0, l_2>0$ are the parameters. Under such a strategy, an extortionate share of expected payoffs larger than $l_1$ and  $l_2$ respectively can be obtained. That is,
\begin{equation}\label{extortion2}
\frac{S_P-l_1}{S_C-l_2} = \chi.
\end{equation}

According to \eqref{extortion2}, the larger $\chi$ is, the more extortionate  share of the expected payoffs the data provider can get. Hence,   $\chi$  is called {\it the extortion factor}.

\begin{theorem}
\label{theorem:exp}
When  $\phi>0$,
\begin{equation}\label{eq:extortion4}
\left\{
\begin{aligned}
\frac{U_P(DD)-l_1}{U_C(DD)-l_2}&\ge\frac{U_P(CC)-l_1}{U_C(CC)-l_2},\\
\frac{U_P(DD)-l_1}{U_C(DD)-l_2}&>1,
\end{aligned}
\right.
\end{equation}
should be satisfied and otherwise,
\begin{equation}\label{eq:extortion5}
\left\{
\begin{aligned}
\frac{U_P(CD)-l_1}{U_C(CD)-l_2}&\ge\frac{U_P(DC)-l_1}{U_C(DC)-l_2},\\
\frac{U_P(CD)-l_1}{U_C(CD)-l_2}&>1
\end{aligned}
\right.
\end{equation}
 needs to hold, so that   the data provider can adopt  the strategy $\mathbf {\hat p}=\phi[(\mathbf U_P-l_1\mathbf 1)-\chi (\mathbf U_C-l_2\mathbf 1)]$   to  enforce the extortionate share of the expected payoffs.
\end{theorem}
\begin{proof}
The strategy to realize the  extortionate relationship, namely \eqref{eq:extortion1},  can be further  expanded as
\begin{displaymath}
\left\{
\begin{aligned}
p_1=&\phi[U_P(CC)-l_1-\chi(U_C(CC)-l_2)]+1\\
e_2p_1+(1-e_2)p_2=&\phi[U_P(CD)-l_1-\chi(U_C(CD)-l_2)]+1\\
p_3&=\phi[U_P(DC)-l_1-\chi(U_C(DC)-l_2)]\\
e_2p_3+(1-e_2)p_4&=\phi[U_P(DD)-l_1-\chi(U_C(DD)-l_2)]\\
\end{aligned}
\right..
\end{displaymath}
Since $p_1, p_2, p_3, p_4, e_1$ and $e_2$ should be within $[0,1]$,  we can get the following inequalities:
\begin{equation} \label{ecl}
\left\{
\begin{aligned}
0&\le\phi[U_P(CC)-l_1-\chi(U_C(CC)-l_2)]+1\le1\\
0&\le\phi[U_P(CD)-l_1-\chi(U_C(CD)-l_2)]+1\le1\\
&0\le\phi[U_P(DC)-l_1-\chi(U_C(DC)-l_2)]\le1\\
&0\le\phi[U_P(DD)-l_1-\chi(U_C(DD)-l_2)]\le1\\
\end{aligned}
\right..
\end{equation}
According to \eqref{ecl}, when $\phi>0$, the following inequality should hold
\begin{equation}\label{eq:extortion2}
\begin{aligned}
\frac{U_P(CC)-l_1}{U_C(CC)-l_2}\le\chi\le\frac{U_P(DD)-l_1}{U_C(DD)-l_2}.
\end{aligned}
\end{equation}
Only when $\chi>1$, can  the data provider obtain an extortionate share of expected payoffs. Hence, combining \eqref{eq:extortion2} with  $\chi>1$, we have \eqref{eq:extortion4}.
Similarly, according to \eqref{ecl},  when $\phi<0$, the following conditions need to satisfy
\begin{equation}\label{eq:extortion3}
\begin{aligned}
\frac{U_P(DC)-l_1}{U_C(DC)-l_2}\le\chi\le\frac{U_P(CD)-l_1}{U_C(CD)-l_2}.\\
\end{aligned}
\end{equation}
Combing the above condition with  $\chi>1$,  we have \eqref{eq:extortion5}.
In summary, when  $\phi>0$,  \eqref{eq:extortion4} should be satisfied and otherwise,  \eqref{eq:extortion5} needs to hold, so that $\chi>1$ and \eqref{ecl} holds, making  the data provider can adopt  the strategy $\mathbf {\hat p}=\phi[(\mathbf U_P-l_1\mathbf 1)-\chi (\mathbf U_C-l_2\mathbf 1)]$   to  enforce the extortionate share of the expected payoffs.
\end{proof}

Finally, we employ the numerical analysis to study the impact of the noises $e_1$ and $e_2$ on  $\chi$ when  $\phi>0$, which further influences whether we can find a feasible $\mathbf {\hat p}$. The 3D graphes in Fig. \ref{fig:extortion}  illustrate how $\chi$  changes with $e_1$ and $e_2$  when $C_P= C_C= 5, C_{P1}= C_{C1}= 2$, and $C_{P2}= C_{C2}=3$, under different settings of $l_1$ and $l_2$. Under each of 3D graphes, we use a 2D graph to show the impact of the noises on the feasible $\chi$, where all the feasible values of $\chi$ are projected on the $e_1-e_2$ plane as  indicated in  the red part of the graph.
In light of \eqref{ecl}, the feasible $\chi$ can lead to the feasible $\mathbf {\hat p}$,  empowering the data provider to execute an extortionate NSZD strategy.

When  the data provider adopts  an extortionate NSZD strategy,  due to $\chi>1$, the expected payoff of the data collector is positive related to that of the data  provider according to \eqref{extortion2}. In another word, if the data collector wants to maximize his expected payoff, he has to maximize that of the data provider. In light of \eqref{eq:p1} and \eqref{eq:p2}, when both players are cooperative,  implying that the state is  $CC$, the payoff of the data provider is the maximum. Hence, as a rational player who is payoff-driven, it is better for the data collector to choose cooperation.  Undoubtedly, when the data collector cooperates, the best action taken by the data provider is also $C$ to pursue  the  maximum payoff according to \eqref{eq:p1} and \eqref{eq:p2}. As a result, taking advantage of  the extortionate NSZD strategy, the data provider can dominate the game to determine the positive relationship between her expected payoff and that of the data collector, compelling the cooperation of the data collector and thus realizing a healthy data trading market where the data provider reports pure data while the collector does not resell data to third parties.

\begin{figure*}
\centering
\subfigure[$l_1=1$, $l_2=2$]
{ \label{fig:subfig:a}
\includegraphics[height=2.8in, width=2.2in]{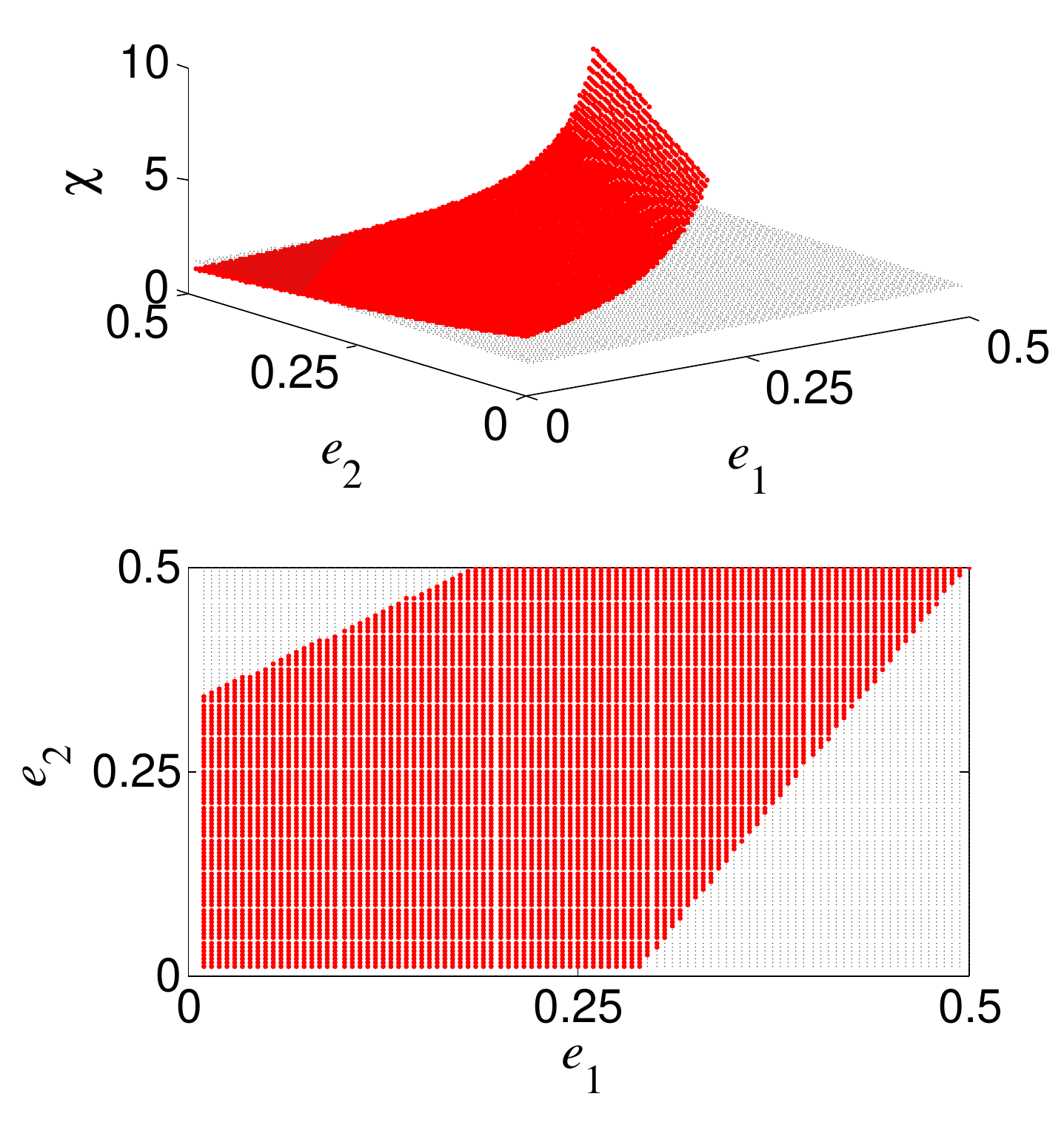}}
\subfigure[$l_1=2$, $l_2=2$]
{ \label{fig:subfig:b}
\includegraphics[height=2.8in, width=2.2in]{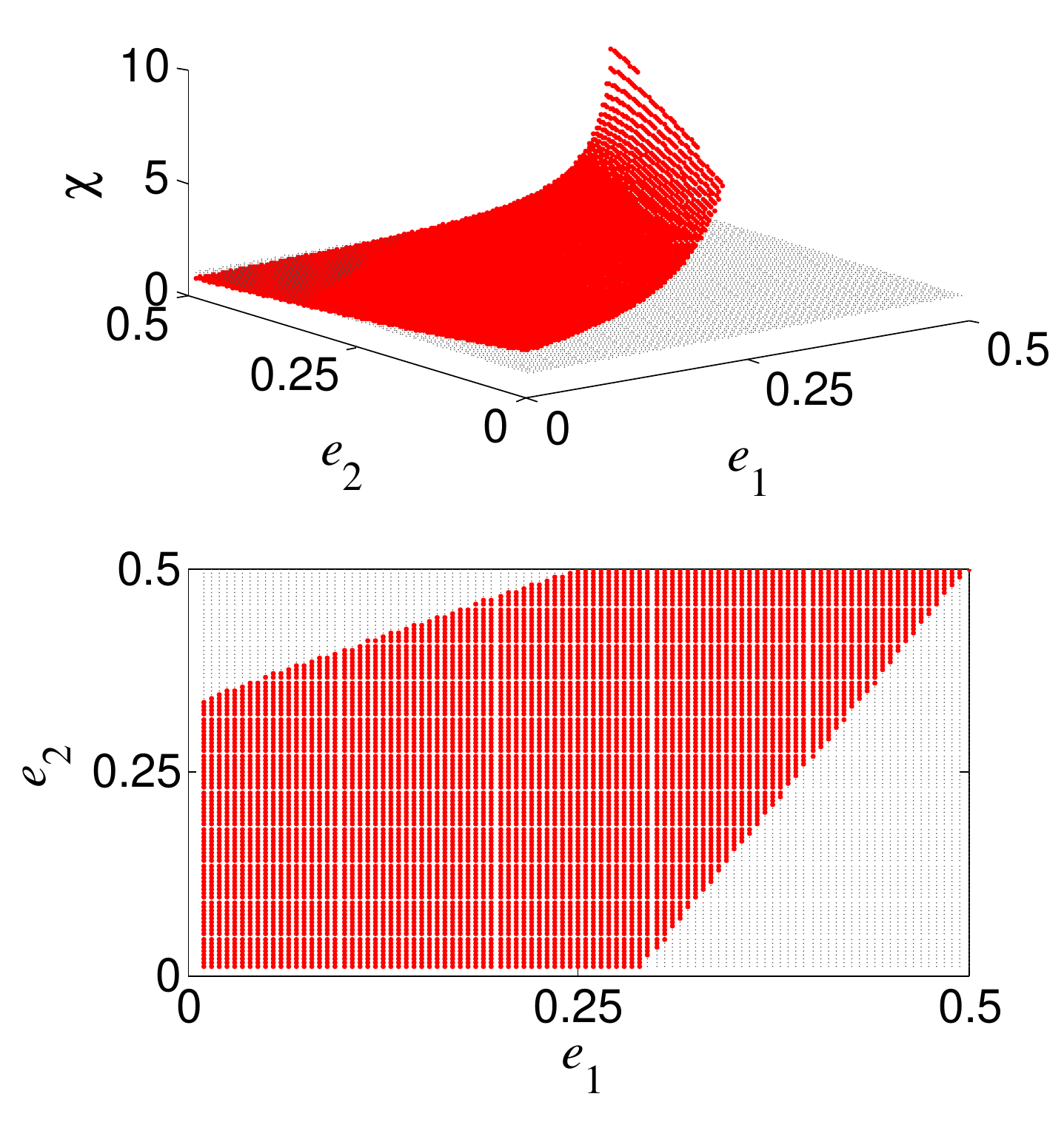}}
\subfigure[$l_1=2$, $l_2=1$]
{ \label{fig:subfig:c}
\includegraphics[height=2.8in, width=2.2in]{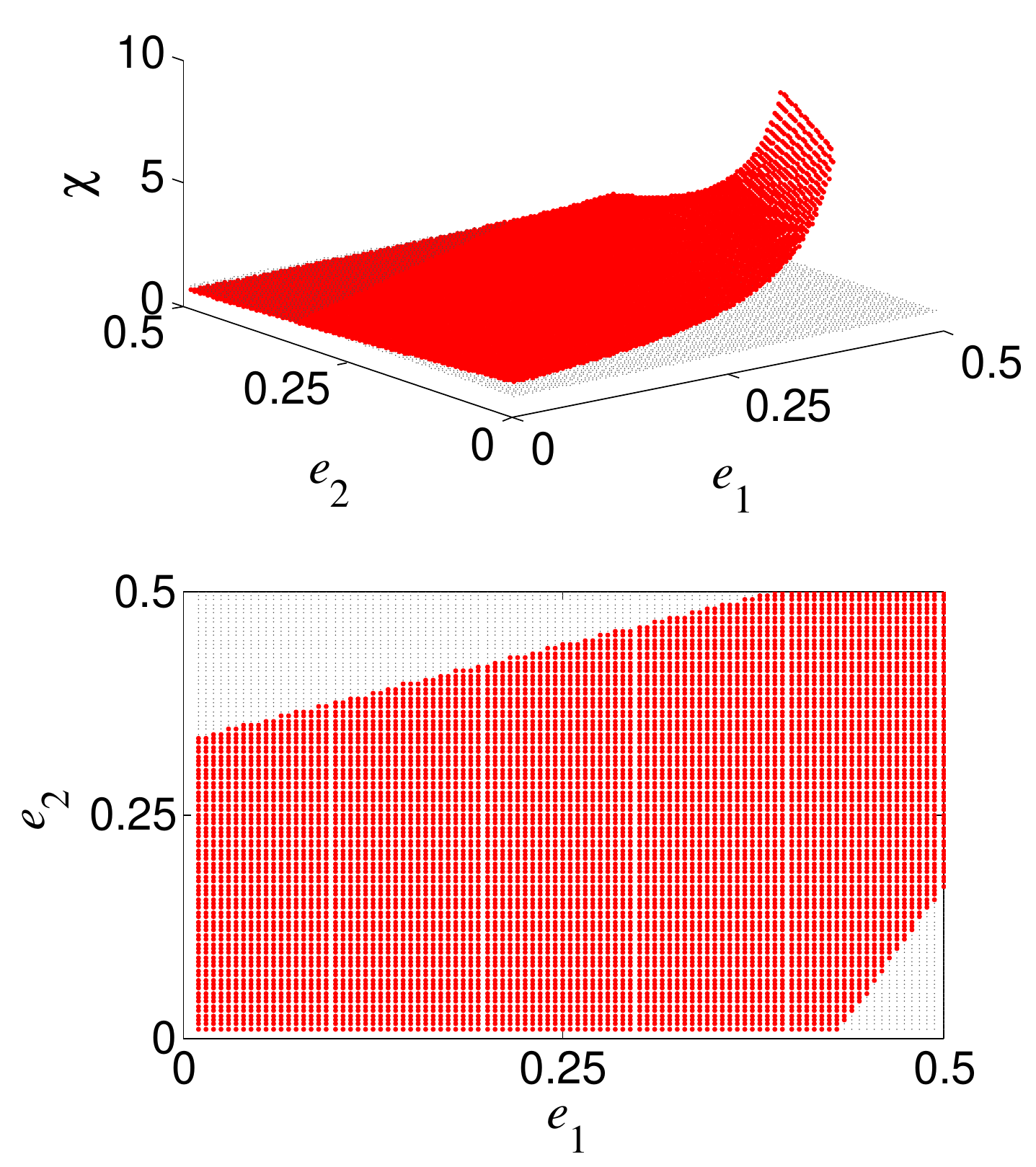}}
\caption{Impact of the noise $e_1$ and $e_2$ on the extortion factor $\chi$ with $C_P=C_C=5$, $C_{P1}=C_{C1}=2$, $C_{P2}=C_{C2}=3$.} \label{fig:extortion}
\end{figure*}


\section{The data collector as an NSZD player}\label{collectorzd}
Using the NSZD strategies, the data provider can reverse her disadvantage to dominate the game with the data collector, through which some incentive mechanisms can be applied to purify the data trading market. Such an idea also raises a concern: what if the data collector, who has already been in a dominant position, adopts the powerful tool, namely the NSZD strategies? To answer this question, in this section, we will analyze whether the data collector has the chance to be an NSZD player. That's to say, whether the data collector's strategy can make $\mathbf {\hat q}=\alpha \mathbf U_P+\beta \mathbf U_C+\gamma \mathbf 1$ hold? Through analysis, we have the following theorems.
\begin{theorem}
\label{theorem:cpinning}
The data collector cannot adopt the pinning  NSZD strategy in the data trading game.
\end{theorem}

\begin{proof}
If the data collector wants to set the opponent's expected payoff by a pinning strategy, $\mathbf {\hat q}=\alpha \mathbf U_P+\gamma$ (set $\beta =0$) should hold, which can be expanded as
\begin{displaymath}
\left\{
\begin{aligned}
0=\alpha U_P(CC)+\gamma\\
0=\alpha U_P(CD)+\gamma\\
(1-e_1)q_1+e_1q_1-1=\alpha U_P(DC)+\gamma\\
(1-e_1)q_1+e_1q_1=\alpha U_P(DD)+\gamma\\
\end{aligned}
\right..
\end{displaymath}
It's obvious that the first two equations of the equation set cannot be satisfied at the same time unless $U_P(CC)=U_P(CD)$, which  deviates from  \eqref{eq:p1}. Hence, this theorem can be proved.
\end{proof}

\begin{theorem}
\label{theorem:cpinning}
The data collector cannot adopt the extortionate NSZD strategy in the data trading game.
\end{theorem}
\begin{proof}
Similarly,  the following equations should be satisfied if the data collector wants to adopt an extortionate strategy.
\begin{small}
\begin{displaymath}
\left\{
\begin{aligned}
0=&\phi[U_P(CC)-l_1-\chi(U_C(CC)-l_2)]+1\\
0=&\phi[U_P(CD)-l_1-\chi(U_C(CD)-l_2)]+1\\
(1-e_1)q_1+e_1q_1-1&=\phi[U_P(DC)-l_1-\chi(U_C(DC)-l_2)]\\
(1-e_1)q_1+e_1q_1&=\phi[U_P(DD)-l_1-\chi(U_C(DD)-l_2)]\\
\end{aligned}
\right..
\end{displaymath}
\end{small}
To make the first two equations hold at the same time,  $\frac{U_P(CC)-l_1}{U_C(CC)-l_2} =\frac{U_P(CD)-l_1}{U_C(CD)-l_2}$ must be satisfied.  However, according to \eqref{eq:p1} and \eqref{eq:c1},  $U_P(CC)>U_P(CD)$ and $U_C(CC)<U_C(CD)$, implying that $\frac{U_P(CC)-l_1}{U_C(CC)-l_2} > \frac{U_P(CD)-l_1}{U_C(CD)-l_2}$. Hence,  we cannot find proper $\phi$ and $\chi$ to satisfy the first two equations above. Therefore, the data collector cannot be an extortioner in the data trading game.
\end{proof}

\section{Conclusion}\label{conclusion}

In this paper, we  study the privacy leakage issue in data trading, by taking a fresh view to make the data provider dominate the game with the collector.  To  that aim,  we set a  noisy sequential game model to depict the interaction between the data provider and its collector where both sides are allowed to take preventive measures  against each other. Based on this model, we propose NSZD strategies, which  empower  the data provider to i)  unilaterally set the expected payoff of the data collector through the pinning strategy; and ii)  enforce a positive relationship between  her and the data collector's expected payoffs by the extortionate strategy. Both strategies make  room for incentive mechanisms to stimulate the cooperation of the data collector. Through the numerical simulations, we examine the impact of key parameters and the feasible region under which the data provider can be an NSZD player. Finally, we prove that the data collector cannot be an NSZD player to take advantage of  the pinning or  extortionate  strategies for  deteriorating privacy leakage in the data trading market.



\bibliographystyle{IEEEtran}
\bibliography{reference}

\begin{thebibliography}{10}
\providecommand{\url}[1]{#1}
\csname url@samestyle\endcsname
\providecommand{\newblock}{\relax}
\providecommand{\bibinfo}[2]{#2}
\providecommand{\BIBentrySTDinterwordspacing}{\spaceskip=0pt\relax}
\providecommand{\BIBentryALTinterwordstretchfactor}{4}
\providecommand{\BIBentryALTinterwordspacing}{\spaceskip=\fontdimen2\font plus
\BIBentryALTinterwordstretchfactor\fontdimen3\font minus
  \fontdimen4\font\relax}
\providecommand{\BIBforeignlanguage}[2]{{%
\expandafter\ifx\csname l@#1\endcsname\relax
\typeout{** WARNING: IEEEtran.bst: No hyphenation pattern has been}%
\typeout{** loaded for the language `#1'. Using the pattern for}%
\typeout{** the default language instead.}%
\else
\language=\csname l@#1\endcsname
\fi
#2}}
\providecommand{\BIBdecl}{\relax}
\BIBdecl

\bibitem{Grindr}
A.~Ghorayshi and S.~Ray, ``Grindr is letting other companies see user hiv
  status and location data,''
  \url{https://www.buzzfeednews.com/article/azeenghorayshi/grindr-hiv-status-privacy},
  buzzFeed. Retrieved April 2, 2018.

\bibitem{Alteryx}
R.~Hackett, ``Data breach exposes 123 million u.s. households,''
  \url{http://fortune.com/2017/12/22/experian-data-breach-alteryx-amazon-equifax/},
  fortune: Time, Inc. Retrieved December 23, 2017.

\bibitem{collector1}
H.~Jin, L.~Su, H.~Xiao, and K.~Nahrstedt, ``Incentive mechanism for
  privacy-aware data aggregation in mobile crowd sensing systems,''
  \emph{IEEE/ACM Transactions on Networking (TON)}, vol.~26, no.~5, pp.
  2019--2032, 2018.

\bibitem{collector4}
L.~Xu, C.~Jiang, Y.~Chen, Y.~Ren, and K.~R. Liu, ``Privacy or utility in data
  collection? a contract theoretic approach,'' \emph{IEEE Journal of Selected
  Topics in Signal Processing}, vol.~9, no.~7, pp. 1256--1269, 2015.

\bibitem{collector2}
L.~K. Fleischer and Y.-H. Lyu, ``Approximately optimal auctions for selling
  privacy when costs are correlated with data,'' in \emph{Proceedings of the
  13th ACM Conference on Electronic Commerce}.\hskip 1em plus 0.5em minus
  0.4em\relax ACM, 2012, pp. 568--585.

\bibitem{collector3}
A.~Ghosh, K.~Ligett, A.~Roth, and G.~Schoenebeck, ``Buying private data without
  verification,'' in \emph{Proceedings of the 15th ACM Conference on Economics
  and Computation}.\hskip 1em plus 0.5em minus 0.4em\relax ACM, 2014, pp.
  931--948.

\bibitem{provider2}
K.~Nissim, S.~Vadhan, and D.~Xiao, ``Redrawing the boundaries on purchasing
  data from privacy-sensitive individuals,'' in \emph{Proceedings of the 5th
  Conference on Innovations in Theoretical Computer Science}.\hskip 1em plus
  0.5em minus 0.4em\relax ACM, 2014, pp. 411--422.

\bibitem{provider1}
W.~Wang, L.~Ying, and J.~Zhang, ``The value of privacy: Strategic data
  subjects, incentive mechanisms and fundamental limits,'' in \emph{ACM
  SIGMETRICS Performance Evaluation Review}, vol.~44, no.~1.\hskip 1em plus
  0.5em minus 0.4em\relax ACM, 2016, pp. 249--260.

\bibitem{spiekermann2015challenges}
S.~Spiekermann, A.~Acquisti, R.~B{\"o}hme, and K.-L. Hui, ``The challenges of
  personal data markets and privacy,'' \emph{Electronic Markets}, vol.~25,
  no.~2, pp. 161--167, 2015.

\bibitem{acquisti2005privacy}
A.~Acquisti and J.~Grossklags, ``Privacy and rationality in individual decision
  making,'' \emph{IEEE security \& privacy}, vol.~3, no.~1, pp. 26--33, 2005.

\bibitem{press2012iterated}
W.~H. Press and F.~J. Dyson, ``Iterated prisoner¡¯s dilemma contains strategies
  that dominate any evolutionary opponent,'' \emph{Proceedings of National
  Academy of Sciences}, vol. 109, no.~26, pp. 10\,409--10\,413, 2012.

\bibitem{noise1}
N.~Mohammed, R.~Chen, B.~Fung, and P.~S. Yu, ``Differentially private data
  release for data mining,'' in \emph{Proceedings of the 17th ACM SIGKDD
  International Conference on Knowledge Discovery and Data Mining}.\hskip 1em
  plus 0.5em minus 0.4em\relax ACM, 2011, pp. 493--501.

\bibitem{noise2}
A.~Friedman and A.~Schuster, ``Data mining with differential privacy,'' in
  \emph{Proceedings of the 16th ACM SIGKDD International Conference on
  Knowledge Discovery and Data Mining}.\hskip 1em plus 0.5em minus 0.4em\relax
  ACM, 2010, pp. 493--502.

\bibitem{anonymity1}
L.~Sweeney, ``k-anonymity: A model for protecting privacy,''
  \emph{International Journal of Uncertainty, Fuzziness and Knowledge-Based
  Systems}, vol.~10, no.~5, pp. 557--570, 2002.

\bibitem{anonymity2}
R.~J. Bayardo and R.~Agrawal, ``Data privacy through optimal k-anonymization,''
  in \emph{Proceedings of the 21st International Conference on Data Engineering
  (ICDE)}.\hskip 1em plus 0.5em minus 0.4em\relax IEEE, 2005, pp. 217--228.

\bibitem{cryptography1}
J.~Sun, X.~Zhu, C.~Zhang, and Y.~Fang, ``Hcpp: Cryptography based secure ehr
  system for patient privacy and emergency healthcare,'' in \emph{Proceedings
  of the 31st International Conference on Distributed Computing Systems
  (ICDCS)}.\hskip 1em plus 0.5em minus 0.4em\relax IEEE, 2011, pp. 373--382.

\bibitem{cryptography2}
K.-H. Lee and P.-L. Chiu, ``An extended visual cryptography algorithm for
  general access structures,'' \emph{IEEE Transactions on Information Forensics
  and Security}, vol.~7, no.~1, pp. 219--229, 2012.

\bibitem{hu2017anti}
Q.~Hu, S.~Wang, L.~Ma, R.~Bie, and X.~Cheng, ``Anti-malicious crowdsourcing
  using the zero-determinant strategy,'' in \emph{Proceedings of the 37th
  International Conference on Distributed Computing Systems (ICDCS)}.\hskip 1em
  plus 0.5em minus 0.4em\relax IEEE, 2017, pp. 1137--1146.

\end{thebibliography}
\end{document}